\documentclass{svproc}
\usepackage{type1cm}        
\usepackage{makeidx}         
\usepackage{graphicx}        
\usepackage{multicol}        
\usepackage[bottom]{footmisc}

\usepackage[shortlabels]{enumitem}
\usepackage{newtxtext}       %
\usepackage[varvw]{newtxmath}  
\DeclareMathOperator{\Ker}{Ker}
\DeclareMathOperator{\Attr}{Attr}

\DeclareMathOperator\Fix{Fix}

\DeclareMathOperator{\spec}{spect}
\usepackage{physics}
\usepackage{hyperref}
\usepackage{cite}
\usepackage{mathrsfs}

\usepackage{url}

\usepackage{color}

\usepackage{hyperref}

\begin{document}
\mainmatter              
\title{Attractor Subspace and Decoherence-Free Algebra of Quantum Dynamics}
%
%
\author{Daniele Amato\inst{1,2} \and Paolo Facchi\inst{1,2} 
 \and Arturo Konderak\inst{1,2,3}}
\authorrunning{Daniele Amato, Paolo Facchi, and Arturo Konderak} 
%
\tocauthor{Daniele Amato, Paolo Facchi, and Arturo Konderak}
\institute{Dipartimento di Fisica, Universit\`a di Bari, I-70126 Bari, Italy \\ \and  
Istituto Nazionale di Fisica Nucleare, Sezione di Bari, I-70126 Bari, Italy \\ \and
Center for Theoretical Physics, Polish Academy of Sciences, \\Aleja Lotnik\'{o}w 32/46, 02-668 Warsaw, Poland \\
D. Amato \\
\email{daniele.amato@ba.infn.it}
\\ P. Facchi\\ \email{paolo.facchi@ba.infn.it}
\\
A. Konderak\\ \email{akonderak@cft.edu.pl}}

\maketitle              

\begin{abstract}
In this review we discuss some results on the asymptotic dynamics of finite-dimensional open quantum systems in the Heisenberg picture. Both the spectral and algebraic approaches to this topic are addressed, with particular emphasis on their relationship. The analysis is conducted in both the  discrete-time and the continuous-time Markovian settings. In the final part of the work, some issues emerging in the infinite-dimensional case are also discussed. In particular, we provide an example of a Markovian evolution whose decoherence-free algebra is a type $III$ von Neumann algebra.
\end{abstract}
\section{Motivation}
\label{intro}
The study of the asymptotic evolution of a system has drawn the attention of the scientific community already in classical mechanics since the seminal work by Poincar{\'e}~\cite{New_methods_Poincare_book}, as solving the problem at all times is quite a formidable task. Not surprisingly, the importance of the analysis at large times appears in the quantum realm too. In  realistic situations a quantum system is always open, i.e. it interacts with its surroundings, usually referred to as the environment or the bath, so that its evolution turns out to be non-unitary. Consequently, the general open-system dynamics is quite more complicated than the closed-system one, and much research has been devoted to understanding the asymptotic behavior of open systems, which resembles the evolution of an isolated system. Moreover, the investigation of what happens only at large times is a crucial question  in physics, e.g. in statistical mechanics~\cite{huangbook2008}. In particular, in recent years the large-time evolution of open quantum systems has turned out to be of upmost importance for applications in quantum computation~\cite{lidar_dec_1998}.

The asymptotic evolution of open quantum systems has been intensively studied by both physicists and mathematicians since the seventies, see e.g.~\cite{evans1977irreducible,Frigerio_Verri_82,Spohn_77,Spohn_rev_1980}. These works  focused on Markovian continuous-time evolutions, called quantum dynamical (or quantum Markov) semigroups and characterized by Gorini, Kossakowski, Lindblad, and Sudarshan (GKLS)~\cite{GKS_76,Lindblad_76}. In this direction, see also the more recent works~\cite{agredo2014decoherence,fagnola_2001,fagnola2019role,Sasso_QMS_2023}, as well as~\cite{carbone2020period,carbone2016irreducible} in the more general context of unital completely positive (UCP) maps. The asymptotic evolution may also be studied in terms of quantum channels in the dynamically equivalent Schr\"odinger picture~\cite{albert2019asymptotics,AF_bounds,AFK_proc,AFK_asympt_1,AFK_asympt4,AFK_asympt3,jex_st_2012,jex_st_2018,rajaramaperipheral_2,rajarama2022peripheral,wolf2010inverse}.

It is not hard to recognize two possible approaches to the study of open-quantum-system asymptotics: a \emph{spectral} approach based on the study of the unimodular eigenvalues and eigenoperators of the UCP map (or of the quantum dynamical semigroup) describing the evolution, and an \emph{algebraic} approach in which the asymptotic dynamics is identified with a $\ast$-homomorphism (or with a semigroup of $\ast$-homomorphisms). Historically, the first approach may be traced back to the works by Jacobs~\cite{jacobs1957fastperiodizitatseigenschaften}, de Leeuw and Glicksberg~\cite{deleeuw1959almost,deleeuw1961applications}. In particular,  in the context of semigroups on Banach spaces, they identified the emergence of a direct-sum decomposition into a reversible algebra related to the peripheral eigenvalues of the semigroup and a stable subspace of asymptotically vanishing vectors.

The second approach is based on the concept of decoherence-free algebra, introduced by Evans~\cite{evans1977irreducible} and whose interest by the quantum probability community~\cite{bulinskii1995some,carbone2020period,fagnola2019role,Frigerio_Verri_82,Sasso_QMS_2023} was renewed by the seminal paper of Blanchard and Olkiewicz~\cite{blanchard2003decoherence}. In particular, in~\cite{blanchard2003decoherence} various physical models were proposed, in which the algebra of observables decomposes into a decoherence-free algebra, on which the semigroup acts homomorphically, and a transient subspace, whose elements vanish in the long-time limit.  
Most works on open-system asymptotics tend to emphasize one of the two schemes, rather than relating them.

The purpose of this Article is to review the connection between the asymptotic subspaces and evolutions considered in the two approaches. The analysis will be mostly performed in a finite-dimensional setting. However, in the last section we will also briefly discuss the issues that arise in the infinite-dimensional case.

The paper is organized as follows: we recall some basic notions on open-system asymptotics in Sect.~\ref{basics_sec}, both in the spectral (Subsect.~\ref{spec-intro_subsec}) and algebraic (Subsect.~\ref{alg-intro_subsec}) approaches. Their connection is then discussed in Sects.~\ref{sec:interplay} and~\ref{Markov_sec}. In particular, Sect.~\ref{sec:interplay} focuses on the discrete-time setting, where the evolution is given by a UCP map, while in Sect.~\ref{Markov_sec} the continuous framework, governed by a GKLS generator, is studied. Finally, Sect.~\ref{sec:infinite_dimensional} dwells on the study of the asymptotic dynamics in the Markovian regime on infinite-dimensional Hilbert spaces. Here, we also give an explicit example of a quantum dynamical semigroup whose decoherence-free algebra is described by a type $III$ von Neumann algebra, which is the type of algebras appearing in quantum field theory.

\section{Preliminary notions}
\label{basics_sec}
		In this paper we will investigate the dynamics in the Heisenberg picture, namely we look at the evolution of the observables. This allows us to properly compare the spectral and algebraic schemes of the asymptotic dynamics.

The dynamics in the unit time is given by a unital completely positive  map $\Phi$ on $\mathcal{B}(\mathcal{H})$~\cite{chruscinski2022dynamical}, the space of linear operators on a finite-dimensional Hilbert space $\mathcal{H}$. Explicitly, the evolved observable $A(n)$ at time $t=n\in \mathbb{N}$ will be the $n$-fold composition $\Phi^n$ of the map $\Phi$ on the observable $A$ at the initial time $t=0$, i.e. $A(n)=\Phi^{n}(A)$. Similarly, in the dynamically equivalent Schr\"odinger picture, the unit-time evolution is described by $\Phi^\dagger$, the adjoint of $\Phi$ with respect to the \emph{Hilbert-Schmidt scalar product}
\begin{equation}
\braket{X}{Y}_{\mathrm{HS}}= \tr(X^\ast Y)\;, \quad X,Y \in \mathcal{B}(\mathcal{H})\;.
\end{equation}
It is easy to see that $\Phi^\dagger$ is a completely positive trace-preserving (CPTP) map on $\mathcal{B}(\mathcal{H})$, called a \emph{quantum channel}, and it describes the evolution of the state of an open quantum system given by a density operator $\varrho$, viz. a positive-semidefinite operator on $\mathcal{H}$ with unit trace.

In the next two subsections, we will present the basic notions behind the spectral and algebraic approaches to open-system asymptotics, whose motivation was already discussed in Sect.~\ref{intro}.

\subsection{Spectral approach to the asymptotics}
\label{spec-intro_subsec}
In this subsection we will discuss the basics of the spectral approach to the asymptotic evolution.
Given a UCP map $\Phi$, we will denote by $\mathrm{spect}(\Phi)$ its spectrum, viz. the set of its eigenvalues. It follows from the definition of $\Phi$ that all its eigenvalues have modulus not larger than one, a property which holds also for the larger class of unital positive maps~\cite{Asorey2008,wolf2012quantum}. Moreover, the spectrum of any UCP map $\Phi$ satisfies the following properties
\begin{enumerate}[a)]
			\item $1 \in \spec(\Phi)$,
			\item $\lambda \in \spec(\Phi) \Rightarrow \bar{\lambda}\in \spec(\Phi)$,
			\item $\spec(\Phi)\subseteq \{ \lambda \in \mathbb{C} \,\vert \, \abs{\lambda} \leqslant 1 \}$.
		\end{enumerate}
We can decompose the spectrum of $\Phi$ as
\begin{equation}
\spec (\Phi) = \spec_{\mathrm{P}} (\Phi)  \cup \spec_{\mathrm{B}} (\Phi)\;,
\label{dec-spec}
\end{equation}
where
\begin{eqnarray}
	\mathrm{spect}_{{\mathrm P}}(\Phi) &\equiv \{ \lambda \in {\mathrm{spect}(\Phi)} \,\vert\, |\lambda|=1 \}\;, \\
\mathrm{spect}_{{\mathrm B}}(\Phi) &\equiv \{ \lambda \in {\mathrm{spect}(\Phi)} \,\vert\, |\lambda| < 1 \}\;, 
\end{eqnarray}
are called the \emph{peripheral} and \emph{bulk} spectra of $\Phi$, respectively. Note that $\spec_{\mathrm{P}}(\Phi) \neq \emptyset$ because of property a), while $\spec_{\mathrm{B}}(\Phi) = \emptyset$ if and only if the UCP map is unitary, i.e. of the form $\Phi(X)=U^\ast X U$, with $X \in \mathcal{B}(\mathcal{H})$ and $U$ unitary on $\mathcal{H}$.

The asymptotic dynamics resulting from the infinite $n$-limit of the evolution takes place inside the asymptotic, peripheral or \emph{attractor} subspace of $\Phi$, defined as
\begin{equation}
	\label{attr_def}\mathrm{Attr}(\Phi) \equiv \mbox{span} \{ X \in \mathcal{B}(\mathcal{H}) \,\vert\, \Phi(X)=\lambda X \mbox{ for some } \lambda \in  \mathrm{spect}_{\mathrm P}(\Phi)\} \;.
\end{equation}
As an immediate consequence of definition~\eqref{attr_def}, the attractor subspace $\Attr(\Phi)$ is invariant under $\Phi$, indeed
\begin{equation}
\label{attr_inva}
\Phi(\Attr(\Phi)) = \Attr(\Phi)\;.
\end{equation}
Importantly, the Jordan decomposition~\cite{kato2013perturbation} of $\Phi$ implies a direct-sum decomposition of the algebra of observables $\mathcal{B}(\mathcal{H})$ into asymptotic and transient components of the evolution, which is related to decomposition~\eqref{dec-spec} of the spectrum. In particular, we have~\cite[Theorem 3.1]{rajaramaperipheral_2}
\begin{equation}
\mathcal{B}(\mathcal{H}) = \Attr(\Phi) \oplus \mathcal{T}(\Phi) \;,
\label{JDK_dec}
\end{equation}
where 
\begin{equation}
\mathcal{T}(\Phi) = \Big\{ X \in \mathcal{B}(\mathcal{H}) \,\Big\vert\, \lim_{n \rightarrow \infty } \Phi^n (X) = 0 \Big\}
\end{equation}
is the \emph{transient part} of the dynamics, involving the generalized eigenoperators corresponding to bulk eigenvalues.
Equation~\eqref{JDK_dec} is an example of Jacobs-de Leeuw-Glicksberg (JdLG) decomposition~\cite{deleeuw1959almost,deleeuw1961applications,jacobs1957fastperiodizitatseigenschaften}. Notice that the rate of convergence to equilibrium is related to the \emph{spectral gap} of the UCP map $\Phi$ defined as
\begin{equation}
g = 1 - \max_{\lambda \in \mathrm{spect}_{{\mathrm B}}(\Phi) } |\lambda|.
\end{equation}
The \emph{peripheral projection} $\mathcal{P}_P$ of $\Phi$ is defined as
\begin{equation}
	\label{P_p_def}
	\mathcal{P}_{\mathrm{P}} =\sum_{\lambda_k\in\mathrm{spect}_{\mathrm{P}}(\Phi)}\mathcal{P}_k\;,
\end{equation}
where $\mathcal{P}_k$ denote the spectral projection of $\Phi$ corresponding to the $k$th peripheral eigenvalue $\lambda_k \in \mathrm{spect}_{\mathrm{P}}(\Phi)$.
It is possible to relate $\mathcal{P}_{\mathrm{P}}$ with the UCP map $\Phi$ via the following formula
\begin{equation}
			\label{P_p_for}
			\mathcal{P}_{\mathrm{P}}=\lim_{i\rightarrow \infty}\Phi^{n_i}\;,
\end{equation}
for some increasing sequence $( n_i)_{i\in\mathbb{N}}$. This ensures that $\mathcal{P}_{\mathrm{P}}$ is a UCP map as well. Clearly, the range of the projection $\mathcal{P}_{\mathrm{P}}$ is the attractor subspace $\mathrm{Attr}(\Phi)$ of $\Phi$. 
An important subspace of $\Attr(\Phi)$ is the \emph{fixed-point subspace} of $\Phi$
\begin{equation}
\Fix(\Phi) = \{ X \in \mathcal{B}(\mathcal{H}) \,\vert\, \Phi(X)=X \}\;, 
\label{fix_def}
\end{equation}
i.e. the space of invariant observables of the system.
Once the attractor subspace has been defined, the asymptotic dynamics is described in our setting by the \emph{asymptotic map} of $\Phi$
\begin{equation}
\hat{\Phi}_{\mathrm{P}} = \Phi \vert_{\Attr(\Phi)}\;.
\end{equation}

By~\eqref{attr_def} the attractor subspace is a vector space. Now, in order to compare it with the decoherence-free algebra of $\Phi$, which we will introduce in Subsect.~\ref{alg-intro_subsec}, it is important to understand whether $\Attr(\Phi)$ is an algebra with respect to the composition product of $\mathcal B(\mathcal H)$. We introduce the following notion, which is equivalent to the notion introduced in~\cite[Definition 1.1]{rajaramaperipheral_2}.
\begin{definition}
\label{per_aut_def}
	Let $\Phi$ be a UCP map. Then $\Phi$ is called peripherally automorphic if and only if $\Attr(\Phi)$ is closed under the composition product, i.e.
	\begin{equation}
		\label{per_aut_eq}
		X\;,\,Y \in \Attr(\Phi) \Rightarrow XY \in \Attr(\Phi)\;.
	\end{equation}
\end{definition} 
\begin{remark}
If $\Phi$ is peripherally automorphic, then $\Attr(\Phi)$ is a $C^\ast$-subalgebra of $\mathcal{B}(\mathcal{H})$.  
\end{remark} 
The name given to this class of maps may be easily justified thanks to the following proposition \cite[Theorem 2.10]{rajaramaperipheral_2}.
\begin{proposition}
			\label{ch_per_aut_1}
			Let $\Phi$ be a UCP map with asymptotic map $\hat{\Phi}_P$. The following conditions are equivalent:
			\begin{enumerate}[a)]
				\item $\Phi$ is peripherally automorphic;
				\item $\hat{\Phi}_{\mathrm{P}}(XY)=\hat{\Phi}_{\mathrm{P}}(X)\hat{\Phi}_{\mathrm{P}}(Y)$ for all $X,Y\in \Attr(\Phi)$.
			\end{enumerate}
		\end{proposition} 
We now temporarily move to the Schr\"odinger picture, where we can define the attractor subspace $\Attr(\Phi^\dagger)$ as well as the fixed-point space $\Fix(\Phi^\dagger)$ of the quantum channel $\Phi^\dagger$ using~\eqref{attr_def} and~\eqref{fix_def} respectively. The assumption of \emph{faithfulness}, introduced in the following definition, is a property concerning $\Fix(\Phi^\dagger)$ which has profound consequences on the asymptotics in both pictures, as we shall see below. 
\begin{definition}
	A UCP map $\Phi$ is said to be faithful if there exists an invertible stationary state $\varrho$ for the CPTP map $\Phi^\dagger$, the Hilbert-Schmidt adjoint of $\Phi$, i.e.
	\begin{equation}
		\Phi^\dagger (\varrho)=\varrho >0\;.
	\end{equation} 
\end{definition} 
Importantly, peripheral automorphism is a weaker condition than faithfulness~\cite[Corollary 2.8]{rajaramaperipheral_2}.
\begin{proposition}
\label{faih_impl_per_aut}
Let $\Phi$ be a faithful UCP map. Then it is peripherally automorphic.
\end{proposition}
\begin{remark}
A peripherally automorphic map $\Phi$ may not be faithful, see~\cite[Example 2.4]{rajaramaperipheral_2}.
\end{remark}
\subsection{Algebraic approach to the asymptotics}
\label{alg-intro_subsec}
In this subsection, we will discuss an alternative approach to study the long-time evolution of an open quantum system. This approach addresses the algebraic properties of the asymptotics. Let us recall the definition and some basic properties of the decoherence-free algebra of $\Phi$, i.e.\ the asymptotic subspace of the dynamics in this approach. The starting point of the analysis is the following: the discrete-time semigroup $(\Phi_n)_{n\in\mathbb{N}}$ describing the evolution of a closed quantum system is made up of $\ast$-automorphisms of $\mathcal{B}(\mathcal{H})$, and since we expect open-system asymptotic dynamics to mimic a closed-system evolution, the asymptotic semigroup should involve $*$-automorphisms of the asymptotic algebra.

The \emph{decoherence-free algebra} of a UCP map $\Phi$ is defined as
\begin{eqnarray}
\mathcal{N} \equiv \{ X \in \mathcal{B}(\mathcal{H}) \,\vert\, \Phi^{n}(Y X) &=& \Phi^{n}(Y)\Phi^{n}(X)\;,\nonumber\\ \Phi^{n}(XY) &=& \Phi^{n}(X)\Phi^{n}(Y) \;,\, \forall Y \in \mathcal{B}(\mathcal{H})\;,\,\forall n \in \mathbb{N} \}\;.
\label{dec_def}
\end{eqnarray}
It is well-known that $\mathcal{N}$ is a ${C}^\ast$-subalgebra of $\mathcal{B}(\mathcal{H})$ invariant under the action of $\Phi$, 
\begin{equation}
	\Phi(\mathcal{N})\subseteq\mathcal{N}\;,
\end{equation}
see e.g.~\cite{carbone2020period}.
The operator Schwarz inequality 
\begin{equation}
\label{op_Schw_ineq}
\Phi(X^\ast X) \geqslant\Phi(X)^\ast \Phi(X)\;, \quad X \in \mathcal{B}(\mathcal{H})\;,
\end{equation}
which holds for any UCP map, implies the following simple characterization of the space $\mathcal{N}$, see~\cite[Theorem 5.4]{wolf2012quantum}.
\begin{proposition}
\label{ch_N}
Let $\Phi$ be a UCP map. Then 
\begin{eqnarray}
	\mathcal{N}=\{ X \in \mathcal{B}(\mathcal{H}) \,\vert\, \Phi^{n}(X^\ast X)
	&=& \Phi^{n}(X)^\ast\Phi^{n}(X)\;,\nonumber\\ \Phi^{n}(XX^\ast ) 
	&=& \Phi^{n}(X)\Phi^{n}(X)^\ast\;,\, \forall n \in \mathbb{N} \}\;.
\end{eqnarray}
\end{proposition}
In order to describe the asymptotic evolution in this setting, let us introduce the \emph{decoherence-free map} of $\Phi$
\begin{equation}
\label{dec-free_map}
\Phi_{\mathcal{N}} = \Phi \vert_{\mathcal{N}}\;,
\end{equation}
and the \emph{decoherence-free semigroup} 
\begin{equation}
\label{dec-free-sem}
\Phi_{\mathcal{N},n}:= \Phi^n \vert_{\mathcal{N}}\;, \quad n \in\mathbb{N}\;.
\end{equation}
 The following corollary is an immediate consequence of the previous results.
\begin{corollary}
\label{largest_N_cor}
Let $\Phi$ be a UCP map. Then its decoherence-free algebra is the largest $C^\ast$-subalgebra $\mathcal{M}$ of $\mathcal{B}(\mathcal{H})$ invariant under $\Phi$ for which the semigroup $(\Phi_{\mathcal{M},n})_{n\in\mathbb{N}}$, defined in~\eqref{dec-free-sem}, is made up of $\ast$-homomorphisms. 
\end{corollary}
\begin{proof}
First, note that the invariance of $\mathcal{M}$ under $\Phi$ implies that  $\Phi_{\mathcal{M},n} = \Phi\vert_{\mathcal{M}}^n$ is a semigroup. 
Take $X\in\mathcal M$, implying $X^*\in\mathcal M$ too, so that for all $n\in\mathbb N$:
\begin{eqnarray}
\Phi^{n}(X^* X) &=& \Phi_{\mathcal{M},n}(X^* X) = \Phi_{\mathcal{M},n}(X)^*\Phi_{\mathcal{M},n}(X) =  \Phi^{n}(X)^*\Phi^{n}(X)\;, 
\\ 
\Phi^{n}(XX^*) &=&  \Phi_{\mathcal{M},n}( X X^*) = \Phi_{\mathcal{M},n}(X)\Phi_{\mathcal{M},n}(X)^* = \Phi^{n}(X)\Phi^{n}(X)^*\ .
\end{eqnarray}
Thus, from Proposition~\ref{ch_N}, $X\in\mathcal N$ and the assertion follows. \qed
\end{proof}
\begin{remark}
\label{Phi_not_auto_rmk}
As already highlighted in~\cite[Remark 1]{carbone2020period}, $\Phi_{\mathcal{N},n}$ is not generally made up of $\ast$-automorphisms.
\end{remark}
\section{Interplay between the spectral and the algebraic framework}
\label{sec:interplay} 
In this section we are going to analyze the connection between the spectral and the algebraic approach. In particular, we aim at investigating the relation between the attractor subspace and the decoherence-free algebra, i.e. the asymptotic subspaces in the two schemes. Before stating the first main result of the section, we need a preparatory lemma~\cite[Lemma 2.4]{hamana1979injective}. 
\begin{lemma}
		\label{lemma_Hamana}
		Let $\Phi$ be an idempotent UCP map, namely $\Phi^2=\Phi$. Then
		\begin{equation}
			\Phi(\Phi(X)\Phi(Y))=\Phi(\Phi(X)Y)=\Phi(X\Phi(Y))\;, \quad X\;,\,Y \in\mathcal{B}(\mathcal{H})\;.
		\end{equation} 
\end{lemma}
Under the faithfulness assumption the attractor subspace reduces to the decoherence-free algebra, and we can say that the spectral and algebraic approaches to the asymptotics become equivalent.
\begin{theorem}
\label{attr_ch_faith}
Let $\Phi$ be a faithful UCP map. Then
\begin{equation}
\Attr(\Phi)=\mathcal{N}\;.
\end{equation}
\end{theorem} 
\begin{proof}
Given $X\in \Attr(\Phi)$, by combining Proposition~\ref{ch_per_aut_1} and Proposition~\ref{faih_impl_per_aut} we have
\begin{equation}
\Phi(X^\ast X)=\Phi(X)^\ast\Phi(X)\;.
\end{equation}
Since $\Attr(\Phi)$ is invariant under $\Phi$ by~\eqref{attr_inva}, we have
\begin{equation}
	\Phi^n(X^\ast X)=\Phi^n(X)^\ast\Phi^n(X)\;, \quad n\in\mathbb{N}\;.
\end{equation}
By exchanging $X \leftrightarrow X^\ast$, we get $\Attr(\Phi)\subseteq \mathcal N $. For the other inclusion, given $X \in \mathcal{N}$ and using~\eqref{P_p_for}
\begin{equation}
\Phi^n(X^\ast X)=\Phi^n(X)^\ast\Phi^n(X) \Rightarrow \mathcal{P}_{{\mathrm{P}}}(X^\ast X)=\mathcal{P}_{\mathrm{P}} (X)^\ast\mathcal{P}_{\mathrm{P}} (X)\;.
\end{equation}
Thus, by applying Lemma~\ref{lemma_Hamana} to $\mathcal{P}_{\mathrm{P}}$ we obtain 
\begin{equation}
\mathcal{P}_{\mathrm{P}} ((X- \mathcal{P}_{\mathrm{P}}(X))^\ast (X-\mathcal{P}_{\mathrm{P}}(X)))=0\;.
\end{equation}
Call $Y = X- \mathcal{P}_{\mathrm{P}}(X)$, and let $\varrho=\Phi^\dagger(\varrho)>0$ be an invertible invariant state for $\Phi^\dagger$. It will be invariant for $\mathcal{P}_\mathrm{P}^\dagger$ too by~\eqref{P_p_for}, and we have
\begin{equation}
\tr(\mathcal{P}_{\mathrm{P}}(Y^\ast Y) \varrho)=\mathrm{tr}(Y^\ast Y \mathcal{P}_{\mathrm{P}}^\dagger(\varrho))=\mathrm{tr}(Y^\ast Y \varrho)=0\;,
\end{equation}
which implies $Y^\ast Y=0$, and so $Y=0$. Thus, $X=\mathcal P_{\mathrm{P}}(X)$, and $\mathcal N\subseteq \Attr(\Phi)$.
\qed
\end{proof}
\begin{corollary}
Let $\Phi$ be a faithful UCP map. Then $\Phi \vert_{\mathcal{N}}$ is a $\ast$-automorphism.
\end{corollary}
\begin{proof}
The assertion follows by combining Theorem~\ref{attr_ch_faith}, claim \emph{b)} of Proposition~\ref{ch_per_aut_1} and the invertibility of $\hat\Phi_\mathrm P$.
\qed
\end{proof}

The relation between the attractor subspace and the decoherence-free algebra may be derived for the larger class of peripherally automorphic UCP maps introduced in Definition~\ref{per_aut_def}, see~\cite[Theorem 2.6]{rajaramaperipheral_2} and~\cite[Remark 10]{fagnola2019role}.
\begin{theorem}
\label{mul_ch_pa_UCP}
Let $\Phi$ be a UCP map. Then, the following conditions are equivalent.
\begin{enumerate}[a)] 
\item $\Phi$ peripherally automorphic;
\item $\Attr(\Phi) \subseteq\mathcal{N}$.
\end{enumerate}
\end{theorem}
\begin{corollary}
Let $\Phi$ be a peripherally automorphic UCP map. Then $\Fix(\Phi) \subseteq \mathcal{N}$.
\end{corollary}
Therefore, if $\Phi$ is not peripherally automorphic, \emph{b)} does not hold and, indeed, it can happen that $\Fix(\Phi) \not \subseteq \mathcal{N}$, see~\cite[Example 3.6]{rajaramaperipheral_2}. 
\section{Markovian evolutions}
\label{Markov_sec}
As discussed in Sect.~\ref{intro}, Markovian continuous-time evolutions were the first ones addressed in the theory of open quantum systems, given their significance in areas such as quantum optics~\cite{breuerpetruccione}. In this section we will first discuss the basic properties of continuous-time evolutions, with particular focus on their asymptotics. Afterwards, we will revisit the connection between the spectral and the algebraic approach to the asymptotic evolution in the continuous setting.

First, recall that a quantum dynamical semigroup (QDS) $(\Phi_t)_{t\in\mathbb{R}^+}$ is a continuous one-parameter semigroup of UCP maps on $\mathcal{B}(\mathcal{H})$. The semigroup property and continuity guarantee the existence of the \emph{generator} $\mathcal{L}$ of the QDS~\cite[Theorem I.3.7]{engelnagelbook}, that is
\begin{equation}
\label{phit_L_rel}
\Phi_t = e^{t\mathcal{L}}\;, \quad t \in \mathbb{R}^+\;.
\end{equation}
Moreover, the generator $\mathcal{L}$ of a QDS may be always written in the \emph{GKLS form}~\cite{GKS_76,Lindblad_76}
\begin{equation}
\label{GKLS}
\mathcal{L}(X)=\mathrm{i}[H,X] + \sum_{k=1}^{d^2-1}  (L_k^\ast X L_k - \frac{1}{2} \{ L_k^\ast L_k , X \})\;,
\end{equation}
with $H=H^\ast$ and $d=\dim\mathcal{H}$.
Thus, it turns out that $\mathcal{L}$ is the sum of two contributions. The first one
\begin{equation}
\mathcal{L}_{\mathrm{H}}(X)=\mathrm{i}[H , X]\;, \quad X\in\mathcal{B}(\mathcal{H})
\end{equation}
is called the \emph{Hamiltonian part} of $\mathcal{L}$, since it is related to a self-adjoint operator $H$ that can be interpreted as an effective Hamiltonian of the system. The second one reads
\begin{equation}
\mathcal{L}_{\mathrm{D}}(X)=\sum_{k=1}^{d^2-1} (L_k^\ast X L_k - \frac{1}{2} \{ L_k^\ast L_k , X \})\;,
\end{equation}
and it is called the \emph{dissipative part} of $\mathcal{L}$. It involves the operators $L_k \in \mathcal{B}(\mathcal{H})$ called noise or \emph{jump operators} and, physically, it describes phenomena arising from the interaction of the system with the environment, like dissipation and decoherence. It vanishes for isolated systems, for which the generator $\mathcal{L}=\mathcal{L}_{\mathrm{H}}$ becomes Hamiltonian. Observe that the GKLS form of $\mathcal{L}$ and, in particular, the decomposition of $\mathcal{L}$ into Hamiltonian and dissipative parts is not unique~\cite{breuerpetruccione}. 

From a spectral perspective, the asymptotics of QDSs is governed by the unimodular eigenvalues of the dynamics, as underlined in the general discrete-time setting in Subsect.~\ref{spec-intro_subsec}. Specifically, given a QDS $(\Phi_t)_{t\in\mathbb{R}^+}$ with generator $\mathcal{L}$, the spectrum of $\mathcal{L}$ satisfies 
\begin{equation}
\spec (\mathcal{L}) \subseteq \{ \lambda \in \mathbb{C} \,\vert\, \Re(\lambda) \leqslant 0 \}\;,
\end{equation}
i.e. it is contained in the left half-plane. Moreover, it is closed under complex conjugation, as is the spectrum of a UCP map, and $0 \in \spec(\mathcal{L})$.
The asymptotic dynamics of the semigroup $(\Phi_t)_{t\in\mathbb{R}^+}$ takes place in the attractor subspace, which in this case is defined as
\begin{equation}
\label{attr-def_QDS}
\Attr (\Phi_t) \equiv \{ X \in \mathcal{B}(\mathcal{H}) \,\vert\, \mathcal{L}(X) = \lambda X \mbox{ for some } \lambda \in \spec_{\mathrm{P}}(\mathcal{L})  \} = \Attr (\Phi)\;.
\end{equation}
Here $\Phi \equiv \Phi_1 =\mathrm e^{\mathcal{L}}$ is the unit-time dynamics, and 
\begin{equation}
\spec_{ P} (\mathcal{L}) = \{  \lambda \in \spec (\mathcal{L}) \,\vert\, \Re(\lambda) = 0 \}
\end{equation}
is the peripheral spectrum of $\mathcal{L}$. Also, the fixed-point subspace of the QDS $(\Phi_t)_{t\in\mathbb{R}^+}$ is defined as
\begin{equation}
\label{fix_QDS}
\Fix(\Phi_t) \equiv \Ker (\mathcal{L}) = \{ X \in\mathcal{B}(\mathcal{H}) \,\vert\, \mathcal{L}(X) = 0 \}\;,
\end{equation}
and it generally differs from the fixed-point subspace $\Fix(\Phi)$ of $\Phi$. More precisely,
\begin{equation}
\Fix(\Phi_t) \subseteq \Fix(\Phi)\;,
\end{equation}
where the inclusion may be strict. As an example, consider a Hamiltonian GKLS generator with a couple of non-zero peripheral eigenvalues being multiples of $2\pi$.

In the Schr\"odinger picture, the dynamics is given by the (Hilbert-Schmidt) adjoint semigroup $(\Phi_t^\dagger)_{t\in\mathbb{R}^+}$ and, again, the study of the asymptotic evolution involves the analysis of the attractor and fixed-point subspaces of the semigroup defined as in \eqref{attr-def_QDS} and \eqref{fix_QDS}. In particular, the latter reads
\begin{equation}
\Fix(\Phi_t^\dagger) = \Ker(\mathcal{L}^\dagger)\;.
\end{equation}
Consistently with the discrete-time setting, we will keep calling stationary the states belonging to $\Fix(\Phi_t^\dagger)$. A notion of faithfulness may now be  introduced for QDSs.
\begin{definition}
	A QDS $(\Phi_t)_{t\in\mathbb{R}^+}$ is said to be faithful if there exists an invertible stationary state $\varrho$ for the adjoint semigroup $(\Phi_t^\dagger)_{t\in\mathbb{R}^+}$, viz. 
	\begin{equation}
		\mathcal{L}(\varrho)=0\;, \quad \varrho >0\;.
	\end{equation} 
\end{definition} 
The asymptotic dynamics is given by the asymptotic semigroup
\begin{equation}
\Phi_{t,\mathrm{P}} = \Phi_t \vert_{\Attr(\Phi_t)}\;.
\end{equation}
From the algebraic viewpoint, the asymptotic subspace is identified with the decoherence-free algebra of the semigroup, given by
\begin{eqnarray}
	\mathcal{N}=\{ X \in \mathcal{B}(\mathcal{H}) \,\vert\, \Phi_t (Y X)
	&=& \Phi_t(Y)\Phi_t(X)\;,\\ \Phi_t(XY) 
	&=& \Phi_t(X)\Phi_t(Y)\;, \,\forall Y \in \mathcal{B}(\mathcal{H})\;,\,\forall t \in \mathbb{R}^+ \}\;.
\end{eqnarray}
Similarly to the discrete-time setting~\eqref{dec-free-sem}, the asymptotic dynamics is described by the decoherence-free semigroup
\begin{equation}
\Phi_{\mathcal{N} , t} = \Phi_t \vert_{\mathcal{N}}\;, \quad t \in\mathbb{R}^+\;.
\end{equation}
The analogues of Proposition~\ref{ch_N} and Corollary~\ref{largest_N_cor} clearly hold in the Markovian continuous-time case.
\begin{remark}
$\Phi_t\vert_{\mathcal{N}}$ is a $\ast$-automorphism for all $t\in\mathbb{R}^+$, since $\Phi_t$ is invertible by \eqref{phit_L_rel}. As discussed in Remark~\ref{Phi_not_auto_rmk}, this is not generally true in the non-Markovian case.
\end{remark}

Once the spectral and algebraic approaches to the Markovian asymptotic dynamics have been introduced, we would like to discuss their relation. First, let us note that the concept of peripheral automorphism introduced in Definition~\ref{per_aut_def} can be naturally extended to the present setting by~\eqref{attr-def_QDS}. Now, the Markovian continuous-time analogues of Theorems~\ref{attr_ch_faith} and~\ref{mul_ch_pa_UCP} are still valid, as stated in the following propositions, see~\cite[Theorem 22]{Sasso_QMS_2023} and~\cite[Remark 10]{fagnola2019role} for the proofs. 
\begin{theorem}
\label{attr_ch_faith_markov}
Let $(\Phi_t)_{t\in\mathbb{R}^+}$ be a faithful QDS. Then
\begin{equation}
\label{eq_Attr_N}
\Attr(\Phi_t)=\mathcal{N}\;.
\end{equation}
\end{theorem}
\begin{remark}
It is possible to see that equation~\eqref{eq_Attr_N} is not a sufficient condition for faithfulness, see~\cite[Example 2]{AFK_asympt3}.
\end{remark}
\begin{theorem}
\label{mul_ch_pa_UCP_mark}
Let $(\Phi_t)_{t\in\mathbb{R}^+}$ be a QDS. Then, the following conditions are equivalent:
\begin{enumerate}[a)]
	\item $(\Phi_t)_{t\in\mathbb{R}^+}$ peripherally automorphic;
	\item $\Attr(\Phi_t) \subseteq\mathcal{N}$.
\end{enumerate}
\end{theorem}
Again, Theorem~\ref{attr_ch_faith_markov} implies the equivalence of the two approaches to the asymptotics for faithful evolutions.
We conclude the section with an important property for the decoherence-free algebra of Markovian evolutions~\cite[Remark 1]{fagnola2017OSID}.
\begin{proposition}
Let $(\Phi_t)_{t\in\mathbb{R}^+}$ be a QDS with generator $\mathcal{L}$ having GKLS form \eqref{GKLS}. Then
\begin{equation}
\mathcal{N} \subseteq \{ X \in \mathcal{B}(\mathcal{H}) \,\vert\, \Phi_t(X) = \mathrm e^{\mathrm{i}tH}X\mathrm e^{-\mathrm{i}tH}\;,\, \forall t\in \mathbb{R}^+ \}\; .
\end{equation}
\end{proposition}
\begin{remark}
The proposition holds for any $H$ in a GKLS representation of $\mathcal{L}$. 
Moreover, equality does not hold in general, see~\cite[Example 8]{fagnola2017OSID}.
\end{remark}
The latter result ensures that for a Markovian continuous-time evolution, the decoherence-free algebra is a subspace in which the evolution of the system is unitary, i.e. asymptotically the quantum system behaves as if it is closed, in line with our physical intuition. In particular, this justifies the idea behind the algebraic approach, viz. finding the subspace in which the dynamics acts as a \emph{group} of $\ast$-automorphisms. 

However, if we relax the Markovianity assumption, the dynamics on the decoherence-free algebra is no longer unitary in general, even in the faithful case. Indeed, general structure theorems of the asymptotics, see~\cite[Theorem 3]{AFK_asympt4},~\cite[Theorem 3.2]{AFK_asympt_1}, and~\cite[Theorem 8]{wolf2010inverse}, reveal the occurrence of an additional (non-unitary) asymptotic evolution described by permutations. The latter disappears for Markovian continuous evolutions, see~\cite[Theorem 5.3]{AFK_asympt_1}. 
\section{Remarks on the infinite-dimensional case}
\label{sec:infinite_dimensional}
In this section we will briefly outline several results regarding open-system asymptotics in the infinite-dimensional Markovian case. Specifically, we will call a quantum dynamical semigroup a family $(\Phi_t)_{t\in\mathbb{R}^+}$ of norm-continuous one-parameter semigroup of UCP normal maps on $\mathcal{B}(\mathcal{H})$, the von Neumann algebra of bounded operators on a (generally infinite-dimensional, but separable) Hilbert space $\mathcal{H}$. Such a semigroup is generated by a GKLS operator in the form~\cite{Lindblad_76}
\begin{equation}
	\mathcal L(X)=\mathrm i[H,X] + \sum_{k\in\mathbb N} (L^*_kXL_k-\frac{1}{2}\{L_k^*L_k,X\})\;, 
\end{equation}
with $X$, $H=H^*$, $(L_k)_{k\in\mathbb N} \in \mathcal{B}(\mathcal{H})$ , such that $\sum_{k\in \mathbb N} L_k^*L_k$ is strongly convergent.

Again, we can follow two strategies for the analysis of the asymptotics: the spectral one seeking a JdLG decomposition of the form \eqref{JDK_dec} and the algebraic one looking at the decoherence-free algebra \eqref{dec_def} of the dynamics. Regarding the first approach, the infinite-dimensional analogue of the attractor subspace of $\Phi_t$ is called the \emph{reversible algebra}, and it reads
\begin{equation}
\mathfrak{M}_{\mathrm{r}} =  \overline{\mathrm{span}} \{ X \in \mathcal{B}(\mathcal{H}) \,\vert\, \mathcal{L}(X)=\lambda X \mbox{ for some } \lambda,	 \mbox{ with } \abs{\lambda}=1\} \;,
\end{equation}
where the closure is taken with respect to the weak-$\ast$ topology. Importantly, $\mathfrak{M}_{\mathrm{r}} = \mathcal{N}$ in the faithful case, see \cite[Theorem 22]{Sasso_QMS_2023}. In this situation, $\mathfrak{M}_{\mathrm{r}}$ and, consequently, $\mathcal N$ are atomic von Neumann algebras (i.e.\ every nonzero projection has a minimal nonzero subprojection) and there exists a JdLG decomposition of $\mathcal{B}(\mathcal{H})$ 
\begin{equation}
\mathcal{B}(\mathcal{H}) = \mathfrak{M}_{\mathrm{r}} \oplus \mathfrak{M}_{s}\;,
\end{equation}
where
\begin{equation}
\mathfrak{M}_{\mathrm{s}} = \{ X \in \mathcal{B}(\mathcal{H}) \,\vert\, 0 \in \overline{\{ \Phi_t(X) \}}_{t\geqslant 0} \}\;,
\end{equation}
is called the \emph{stable subspace} of the semigroup. Therefore, also for infinite-dimensional open systems, the spectral and algebraic approaches turn out to be equivalent for faithful Markovian evolutions. 

For non-faithful Markovian evolutions interesting issues arise in the infinite-dimen\-sional case: the decoherence-free algebra $\mathcal{N}$ can be non-atomic, and also of type different from $I$. In the following example we will construct a decoherence-free algebra which is a factor of type $III$, generalizing the construction discussed in~\cite[Subsect. II.A]{Sasso_QMS_2023}. These algebras are the most singular ones in the original classification by von Neumann~\cite{onringoperator}. In particular, all the projections on such factor are similar, and their extremely singular nature made them less attractive in von Neumann's seminal papers~\cite{onringoperatorii,onringoperatoriii}. Nevertheless, they turned out to be extremely important in the algebraic description of quantum field theory~\cite{Araki1964,Haag1996,Kadison1963}. Interestingly, they can also be obtained in the alternative description of quantum mechanics based on the Schwinger picture~\cite{CIAGLIA2023104901}, where the  building blocks of the theory are the transitions, rather than the states or the observables~\cite{ciaglia_categorical,groupoidI,groupoidII,Schwinger}.

\begin{example}[Type $III$ decoherence-free algebra]
	In this example, we will construct a quantum dynamical semigroup whose decoherence-free algebra is a factor of type $III$. In the standard approach, one considers the action of a countable group on an Abelian von Neumann algebra of operators, and then takes the von Neumann algebra generated by operators representing such action on a proper Hilbert space. To our purpose, we will simply identify a countable subfamily of the generators with the jump operators of the GKLS generator. Specifically, we will employ the construction discussed in Appendix~\ref{app:ppfactors} and hereafter summarized. 
	
	Let $\Omega_\infty$ be the set of binary sequences, and  consider the group $\mathcal{G}= (\Omega , \oplus )$, where $\Omega \subseteq \Omega_\infty $ is the set of finite sequences [which are eventually zero, see~\eqref{Omega_def}], and the composition law $\oplus$ is the entrywise sum modulo 2. Moreover, consider the probability measure $\mu_\lambda$ on the binary sequences which on the $n$th component $x_n$ assigns a probability $\lambda$ if $x_n=0$, and $1-\lambda$ if $x_n=1$, with $\lambda$ a fixed parameter in $(0,1/2]$. Also, let $\mathcal H$ be the Hilbert space of complex-valued $\mu_\lambda$-measurable square-integrable functions  on $\Omega_\infty\times \Omega$, defined in~\eqref{eq:hilbert_space_pukanski}. 
	
	Let $e_k=(\delta_{jk})_{j\in\mathbb N}\in\Omega$ and consider the jump operators
	\begin{equation}
		M_k=V_{e_k}\;,
	\end{equation}
	with $V_{e_k}$ being the unitary representative on $\mathcal{H}$ of the group composition with $e_k$, see~\eqref{eq:unitary_representation}. We also need to consider the following family of jump operators on $\mathcal{H}$ in the form
	\begin{equation}
		N_k = L_{\psi_k}\;,
	\end{equation}
	with $L_\phi$ being the multiplication operator by $\phi$, defined as in~\eqref{eq:multiplication_repr}, and $\psi_k$ of the form
	\begin{equation}
		\psi_k(x) =
		\begin{cases}
			\ +1 \mbox{ for } x_k = 0, \\
			\ -1 \mbox{ for } x_k = 1.
		\end{cases}\;
	\end{equation}
	We can now construct the GKLS generator as
	\begin{equation}
	\label{L_III_def}
		\mathcal L_\lambda(X) = \sum_{k\in \mathbb N} m_k (M_k^* X M_k - X)+\sum_{k\in \mathbb N} n_k (N_k^* X N_k - X)\;,
	\end{equation}
	with $(m_k)_{k\in\mathbb{N}}$ and $(n_k)_{k\in\mathbb{N}}$ two strictly positive sequences with $\sum_k m_k<+\infty$, $\sum_k n_k<+\infty$.
	We claim that the decoherence-free algebra of the QDS $e^{t\mathcal{L}_\lambda}$ is non-atomic. Indeed, we have the following result.
\begin{proposition}
\label{prop: typeIII}
The decoherence-free algebra 
of the quantum dynamical semigroup generated by 
$\mathcal L_\lambda$ in~\eqref{L_III_def}
is a type $III$ factor for all $\lambda\in(0,1/2)$, and a  type $II_1$ factor for $\lambda=1/2$.
\end{proposition}
\begin{proof}
The decoherence-free algebra 
is given by
	\begin{equation}
		\mathcal N=\{M_k,N_k\, \vert \, k\in\mathbb N\}'\;,
	\end{equation}
	as a consequence of \cite[Proposition 2]{fagnola2019role}.
	Here, note that $M_k$ and $N_k$ are both self-adjoint and unitary operators on $\mathcal{H}$.
		
Since the type of an algebra is the same as the type of its commutant~\cite[Corollary 2.24]{takesaki2013theory}, we only need to show that the von Neumann algebra generated by the jump operators is exactly the Puk{\'a}nszky-Powers factor. Call the latter algebra $\mathcal M$, so that $\mathcal M=\mathcal N'$. First, observe that every $x^o\in\Omega$ is a finite sum of some $e_k$, so every $V_{x^o}\in \mathcal M$, as $V$ is a unitary representation of the group $\mathcal G$.
	
	Furthermore, we can prove that all $L_\phi$ defined by~\eqref{eq:multiplication_repr} belong to $\mathcal M$. First, the cylinders defining the $\sigma$-algebra over $\Omega_\infty$ are in the form
	\begin{equation}
		\label{eq:cylinder}
		\mathscr C(y_1,\dots,y_n) = \{x\in\Omega_\infty\,|\, x_1=y_1,\dots,x_n=y_n\} = \bigcap_{k=1}^n\{x\in\Omega_\infty\,|\, x_k=y_k\} \equiv \bigcap_{k=1}^n \mathscr{C}_k\;,
	\end{equation}
	where $(y_1,\dots,y_n)\in\{0,1\}^n$. The characteristic function over $\mathscr{C}_k$ in~\eqref{eq:cylinder} is 
	\begin{equation}
	\mathbf{1}_{\mathscr{C}_k}(x)=(1-y_k) \psi_k(x) - y_k\psi_k(x)\;.
	\end{equation}
	Therefore, the characteristic function $\mathbf{1}_{\mathscr{C}}$ over the cylinder $\mathscr C \equiv \mathscr C(y_1,\dots,y_n)$ reads
	\begin{equation}
	\mathbf{1}_{\mathscr C} = \prod_{k=1}^n \mathbf{1}_{\mathscr{C}_k}\;.
	\end{equation} 
	
	As a result, the multiplication operators  $L_{\mathbf{1}_{\mathscr{C}_k}}$ and $L_{\mathbf{1}_{\mathscr C}}$, 
	belong to $\mathcal M$. As the span of the characteristic functions $\mathbf{1}_{\mathscr{C}}$ of all cylinders is dense in $\mathrm L^\infty(\Omega_\infty,\mu_\lambda)$ in the $*$-weak topology, 
	we can conclude that $\mathcal M$ is indeed the Puk{\'a}nszky-Powers algebra. 
	
	In particular, $\mathcal{M}$, as well as the decoherence-free algebra $\mathcal N$, is a type $III$ factor whenever $0<\lambda<1/2$, and a  type $II_1$ factor for $\lambda=1/2$~\cite{pukanski_some}. \qed	
\end{proof}

\end{example}

\section{Conclusions}
In the present work, we provided an updated short overview of the results regarding open-system Heisenberg asymptotics in the finite-dimensional case. We highlighted the existence of two kinds of strategies to address this topic: a spectral approach related to the peripheral (unimodular) eigenvalues of the dynamics, and an algebraic one where the asymptotic evolution is characterized by a $\ast$-homomorphism. 

After introducing the basic notions of the two strategies, we presented the main connections between the two approaches in literature. Specifically, we examined the relation between the attractor subspace and the decoherence-free algebra, which are the asymptotic subspaces corresponding to the spectral and algebraic approaches, respectively. We illustrated how faithfulness implies the equality of such spaces and, consequently, the equivalence of the two methods, unveiling the strength of this assumption. Additionally, the implications of the milder property of peripheral automorphism were discussed. 

Finally, we commented on the infinite dimensional case, where new issues arise for both the algebraic and topological viewpoints. In particular, we were able to construct a quantum dynamical semigroup whose decoherence-free algebra is non-atomic and is given by a type $III$ factor.

This short review aims to bridge the gap between the mathematical and physical communities, thereby enhancing the understanding of the asymptotic problem, both conceptually and technically.

\section{Acknowledgements}
We acknowledge support from INFN through the project ``QUANTUM'', from the Italian National Group of Mathematical Physics (GNFM-INdAM), and from the Italian funding within the ``Budget MUR - Dipartimenti di Eccellenza 2023--2027''  - Quantum Sensing and Modelling for One-Health (QuaSiModO).\ PF acknowledges support from Italian PNRR MUR project CN00000013 -``Italian National Centre on HPC, Big Data and Quantum Computing''.\ DA acknowledges support from PNRR MUR project PE0000023-NQSTI.\ AK acknowledges the support from the ``Young Scientist" project at the Center for Theoretical Physics, Polish Academy of Sciences.
\appendix
	\section{Puk\'anszky-Powers factors}
	\label{app:ppfactors}
	In this appendix, we summarize some aspects regarding the construction of type $II$ and $III$ factors. This approach was originally discussed by von Neumann~\cite{onringoperatoriv}, and subsequently by Puk\'anszky and Powers~\cite{Powers1967,pukanski_some}, who used it to identify inequivalent classes of type $III$ algebras. The construction is based on the ergodic and free actions of countable groups on Abelian algebras~\cite{takesaki2013theory}, but it can be also framed in the context of infinite products of Hilbert spaces~\cite{von1939infinite}. By defining a proper convolution between elements, one defines a crossed product algebra, and the type of the algebra is completely characterized by the structure of the underlying Abelian algebra~\cite[Theorem~V.7.12]{takesaki2013theory}.

	We will not give the general theorem, but we will limit to present the construction from Puk{\'a}nszky~\cite{pukanski_some}, see also~\cite{CIAGLIA2023104901} for technical details.
	In particular, we consider the set of binary sequences 
	\begin{equation}
		\Omega_\infty=\bigl\{(x_n)_{n\in\mathbb N}\,|\,x_n\in\{0,1\}\bigr\}\;,
	\end{equation}
	equipped with the product $\sigma$-algebra $\Sigma$ generated by the cylinders defined in \eqref{eq:cylinder}. Then the measurable space $(\Omega_\infty , \Sigma )$ may be equipped with the following probability measure. On the $n$th component $x_n$, a probability $\lambda$ is assigned if $x_n=0$, and $1-\lambda$ if $x_n=1$, with $\lambda$ a fixed parameter in $(0,1/2]$. Consequently, the probability measure $\mu_\lambda$ on $\Omega_\infty$ is constructed through Kolmogorov theorem, see~\cite[Theorem II.3.3]{Shiryaev1996}.
	The Abelian algebra underlying the factor is thus $\mathrm L^\infty(\Omega_\infty,\mu_\lambda)$. Afterwards, a countable group is introduced as the set $\Omega$ of finite sequences of the form
	\begin{equation}
	\label{Omega_def}
		\Omega = \{(x_k)\in\Omega_\infty\,|\, x_k=0\;,\, k>N\;,\textnormal{ for some }N\in\mathbb{N} \}\;.
	\end{equation}
	Using von Neumann's notation~\cite{von1939infinite}, we will call $x,y,z$ the elements of $\Omega_\infty$, and with $x^o,y^o,z^o$ those of $\Omega$. The composition rule on $\Omega_\infty$ and, consequently, $\Omega$ is given by the entry-wise sum modulo $2$. More explicitly, given $x^o=(x_k)_{k\in\mathbb N}$, $y^o=(y_k)_{k\in\mathbb N} \in \Omega$, the composition law reads
	\begin{equation}
		x\oplus y=(x_k\oplus y_k)_{k\in\mathbb N}\;,
	\end{equation}
	where, with a slight abuse of notation, $\oplus$ also denotes the sum modulo $2$ in $\{ 0,1 \}$. The algebra $\mathrm L^\infty(\Omega_\infty,\mu_\lambda)$ and the group $\mathcal G \equiv (\Omega , \oplus)$ can be represented on the Hilbert space
	\begin{equation}
		\label{eq:hilbert_space_pukanski}
		\mathcal  H=\Bigg\{F:\Omega_\infty\times \Omega\rightarrow\mathbb C \,\Bigg| \,\sum_{x^o\in\Omega}\int_{\Omega_\infty} |F(x,x^o)|^2d\mu_\lambda(x)<+\infty\Bigg\}\;.
	\end{equation}
	In formulas, given $\phi \in \mathrm L^\infty(\Omega_\infty,\mu_\lambda)$ and $y^o \in \Omega$, let us define
	\begin{eqnarray}
		(L_\phi F)(x,x^o)&=&\phi(x) F(x,x^o)\;,\label{eq:multiplication_repr}\\
		(V_{y^o}F)(x,x^o) &=& \sqrt{\frac{d\mu_{\lambda}^{y^o}}{d\mu_{\lambda}}}(x)F(x\oplus y^o,x^o\oplus y^o)\;.\label{eq:unitary_representation}
	\end{eqnarray}
	Here, $d\mu_{\lambda}^{y^o}/d\mu_{\lambda}$ is the Radon-Nikodym derivative of the translation $\mu_{\lambda}^{y^o}$ of the measure $\mu_{\lambda}$ by $y^o$ with respect to $\mu_{\lambda}$. In particular, $V$ is a unitary representation of $\mathcal G$, and we can consider the von Neumann algebra $\mathcal M_{\mathrm P}$ generated by the operators $L_\phi$ and $V_{y^o}$. Then it turns out that:
	\begin{enumerate}[a)]
		\item When $\lambda = 1/2$ the algebra $\mathcal M_{\mathrm P}$ is the hyperfinite type $II_1$ factor;
		\item When $0<\lambda<1/2$ we obtain type $III$ factors. In particular, these algebras are inequivalent von Neumann algebras~\cite{Powers1967}, and are indicated as $\mathcal R_{\tilde{\lambda}}$, with $\tilde{\lambda}=\lambda/(1-\lambda)$.
	\end{enumerate}
	As a final remark, observe that the crossed product algebra of the group $\Omega$ over the algebra $\mathrm L^\infty(\Omega_\infty,\mu_\lambda)$ with unitary action $V$ is given by the commutant of $\mathcal M_{\mathrm P}$~\cite{takesaki2013theory}.

\bibliography{biblioproc.bib}
\bibliographystyle{spmpsci}
\end{document}